\documentclass[aps,11pt,twoside]{revtex4}
\usepackage{amsmath,latexsym,amssymb,verbatim,enumerate,graphicx}

\usepackage{hyperref}
\usepackage{color}

\usepackage{theorem}
\newtheorem{definition}{Definition}[section]
\newtheorem{proposition}[definition]{Proposition}
\newtheorem{lemma}[definition]{Lemma}

\newtheorem{theorem}{Theorem}

\def\squareforqed{\hbox{\rlap{$\sqcap$}$\sqcup$}}
\def\qed{\ifmmode\squareforqed\else{\unskip\nobreak\hfil
\penalty50\hskip1em\null\nobreak\hfil\squareforqed
\parfillskip=0pt\finalhyphendemerits=0\endgraf}\fi}
\def\endenv{\ifmmode\;\else{\unskip\nobreak\hfil
\penalty50\hskip1em\null\nobreak\hfil\;
\parfillskip=0pt\finalhyphendemerits=0\endgraf}\fi}
\newenvironment{proof}[1][Proof]{\noindent \textbf{{#1~} }}{\qed}
\newcommand{\bra}[1]{\langle #1|}
\newcommand{\ket}[1]{|#1\rangle}
\newcommand{\braket}[2]{\langle #1|#2\rangle}
\newcommand{\tr}{\text{tr}}

\newcommand{\id}{\mathbb{I}}

\mathchardef\ordinarycolon\mathcode`\:
\mathcode`\:=\string"8000
\def\vcentcolon{\mathrel{\mathop\ordinarycolon}}
\begingroup \catcode`\:=\active
  \lowercase{\endgroup
  \let :\vcentcolon
  }

\newcommand{\nc}{\newcommand}
\nc{\rnc}{\renewcommand} \nc{\beq}{\begin{equation}}
\nc{\eeq}{{\end{equation}}} \nc{\bea}{\begin{eqnarray}}
\nc{\eea}{\end{eqnarray}} \nc{\beqa}{\begin{eqnarray}}
\nc{\eeqa}{\end{eqnarray}} \nc{\lbar}[1]{\overline{#1}}
\nc{\conv}{\operatorname{conv}}
\nc{\smfrac}[2]{\mbox{$\frac{#1}{#2}$}} \nc{\Tr}{\operatorname{Tr}}
\nc{\ox}{\otimes} \nc{\dg}{\dagger} \nc{\dn}{\downarrow}
\nc{\lmax}{\lambda_{\text{max}}}
\nc{\lmin}{\lambda_{\text{min}}}

\nc{\csupp}{{\operatorname{csupp}}}
\nc{\qsupp}{{\operatorname{qsupp}}} \nc{\var}{\operatorname{var}}
\nc{\rar}{\rightarrow} \nc{\lrar}{\longrightarrow}
\nc{\poly}{\operatorname{poly}}
\nc{\polylog}{\operatorname{polylog}} \nc{\Lip}{\operatorname{Lip}}
\nc{\mb}[1]{\mathbf{#1}}
\nc{\ep}{\epsilon}
\nc{\Om}{\Omega}
\nc{\wt}[1]{\widetilde{#1}}

\def\>{\rangle}
\def\<{\langle}

\nc{\glneq}{{\raisebox{0.6ex}{$>$}  \hspace*{-1.8ex} \raisebox{-0.6ex}{$<$}}}
\nc{\gleq}{{\raisebox{0.6ex}{$\geq$}\hspace*{-1.8ex} \raisebox{-0.6ex}{$\leq$}}}

\nc{\RR}{{{\mathbb R}}}
\nc{\FF}{{{\mathbb F}}}
\nc{\HH}{{{\mathbb H}}}
\nc{\NN}{{{\mathbb N}}}
\nc{\ZZ}{{{\mathbb Z}}}
\nc{\PP}{{{\mathbb P}}}
\nc{\QQ}{{{\mathbb Q}}}
\nc{\UU}{{{\mathbb U}}}
\nc{\WW}{{{\mathbb W}}}
\nc{\EE}{{{\mathbb E}}}
\rnc{\SS}{{{\mathbb S}}}
\nc{\vholder}[1]{\rule{0pt}{#1}}
\nc{\wh}[1]{\widehat{#1}}
\nc{\h}[1]{\widehat{#1}}

\nc{\ob}[1]{#1}

\def\beq{\begin {equation}}
\def\eeq{\end {equation}}

\def\be{\begin{equation}}
\def\ee{\end{equation}}

\nc{\eq}[1]{Eq.~(\ref{eq:#1})} \nc{\eqs}[2]{Eqs.~(\ref{eq:#1}) and
(\ref{eq:#2})}

\nc{\eqn}[1]{Eq.~(\ref{eqn:#1})}
\nc{\eqns}[2]{Eqs.~(\ref{eqn:#1}) and (\ref{eqn:#2})}

\nc{\region}{\cS\cW}

\begin{document}

\title{{\Large On Hastings' counterexamples to the minimum output entropy additivity conjecture}}

\author{Fernando G.S.L. Brand\~ao}
\email{fernando.brandao@imperial.ac.uk}
\affiliation{Institute for Mathematical Sciences, Imperial
College London, London SW7 2BW, UK}
\affiliation{QOLS, Blackett Laboratory, Imperial College
London, London SW7 2BW, UK}

\author{Micha\l{} Horodecki}
 \email{fizmh@ug.edu.pl}
\affiliation{Institute for Theoretical Physics and Astrophysics, University of Gda\'nsk, 80-952 Gda\'nsk, Poland}

%\date{18 October 2006}

\begin{abstract}
Hastings recently reported a randomized construction of channels violating the minimum output entropy additivity conjecture. Here we revisit his argument, presenting a simplified proof. In particular, we do not resort to the exact probability distribution of the Schmidt coefficients of a random bipartite pure state, as in the original proof, but rather derive the necessary large deviation bounds by a concentration of measure argument. Furthermore, we prove non-additivity for the overwhelming majority of channels consisting of a Haar random isometry followed by partial trace over the environment, for an environment dimension much bigger than the output dimension. This makes Hastings' original reasoning clearer and extends the class of channels for which additivity can be shown to be violated. 
\end{abstract}

\maketitle

\parskip .75ex

%%%%%%%%%%%%%%%%%%%%%%%%%%%%%%%%%%%%%%%%%%%%%%%%%%%%%%%%%%%%%%%%%%%%%%%%

\section{Introduction}

The oldest problem in quantum information theory is probably the determination of the capacity of a quantum-mechanical channel for classical information transmission. Given a quantum channel from a sender to a receiver, characterized by a trace preserving completely positive map ${\cal E}$, its classical capacity is defined as the maximum number of bits which can be reliably sent per use of the channel, in the limit of infinitely many realizations of the channel. Holevo \cite{Hol98} and Schumacher-Westmoreland \cite{SW97} proved the following formula for the classical information transmission capacity:    
\begin{equation} \label{capacity}
C({\cal E}) = \chi^{\infty}({\cal E}) := \lim_{n \rightarrow \infty} \frac{\chi^{\infty}({\cal E}^{\otimes n})}{n},
\end{equation}
where the Holevo $\chi$-quantity \cite{Hol73} is defined by
\begin{equation}
\chi({\cal E}) := \max_{\{ p_i, \rho_i \}} S\left( {\cal E}\left( \sum_i p_i \rho_i  \right)\right) - \sum_i p_i S \left( {\cal E}\left( \rho_i \right)\right),
\end{equation}
with $S$ being the von Neumann entropy and the maximization ranging over all ensembles $\{ p_i, \rho_i \}$.

An important question concerning the capacity formula given by Eq. (\ref{capacity}) is whether the \textit{regularization} of the $\chi$ quantity to infinitely many uses of the channel is really needed in the right-hand-side of Eq. (\ref{capacity}). Indeed, such necessity would render the evaluation of the formula given by Eq. (\ref{capacity}) in general intractable; moreover, it would show that we do not fully understand the structure of the optimal coding strategy, since from Eq. (\ref{capacity}) we cannot say anything about the - in general entangled - states $\rho_i$ appearing in the optimal ensemble. On a more positive note, the need of regularization would also show that we can boost the information transmission capacity by using entangled encoding states.  

Based on numerical evidence in low dimensions and several results for particular classes of channels (e.g. \cite{AHW00, King02, Sho02, King03, 
DHS04, MY04, Alicki04, Dat04, AF05, KMNR05, DR05, WE05}), it was conjectured the $\chi$-quantity is additive, i.e. for every pair of channels ${\cal E}_1$, ${\cal E}_2$,
\begin{equation} \label{fulladditivity}
\chi({\cal E}_{1} \otimes {\cal E}_2) = \chi({\cal E}_1) + \chi({\cal E}_2).
\end{equation}
The validity of this conjecture would imply that the classical capacity of a quantum channel is given simply by its Holevo $\chi$-quantity, which would constitute a \textit{single-letter} formula for the capacity. It turns out that Eq. (\ref{fulladditivity}) is in fact equivalent to the to the non-necessity of the limit in Eq. (\ref{capacity}) \cite{FW07}: $C({\cal E}) = \chi({\cal E})$ for every channel ${\cal E}$ if, and only if, Eq. (\ref{fulladditivity}) holds true for every pair of channels ${\cal E}_1$, ${\cal E}_2$ (see also \cite{BHPV07}).

The additivity of the $\chi$-quantity can be related to other additivity questions. The first concerns the entanglement cost of a bipartite quantum state $\rho$ shared by Alice and Bob. It is defined as the optimal rate of EPR pairs needed for the formation of $\rho$, in the limit of asymptotically many copies of the state, by local operations and classical communication between Alice and Bob. It was shown in \cite{HHT01} that the entanglement cost is given by 
\begin{equation} \label{cost}
E_C(\rho) := \lim_{n \rightarrow \infty} \frac{E_F(\rho^{\otimes n})}{n},
\end{equation}
where the entanglement of formation \cite{BVSW96} is defined as
\begin{equation}
E_F(\rho) := \min_{\{ p_i, \rho_i \}} \sum_i p_i S\left(\tr_A \left(\ket{\psi_i}\bra{\psi_i} \right) \right),
\end{equation}
with the minimization taken over all pure state ensembles of $\rho$. As shown by Shor in Ref. \cite{Sho04} (building on \cite{MSW04, AB04, Pom03}), the additivity of the entanglement of formation is equivalent to the additivity of $\chi$ as given by Eq. (\ref{fulladditivity}).

The second additivity question concerns the distillable common randomness of a bipartite state, given by the optimal rate of maximally correlated classical bits that can be extracted from a bipartite quantum state, when classical communication is allowed from Alice to Bob (the rate of bits communicated being subtracted from the rate of maximally correlated bits obtained in the end of the protocol). Devetak and Winter proved that \cite{DW03}
\begin{equation} \label{distillablernadomness}
C_D^{\rightarrow}(\rho) := \lim_{n \rightarrow \infty} \frac{I^{\rightarrow}(\rho^{\otimes n})}{n},
\end{equation}
with 
\begin{equation} \label{HV}
I^{\rightarrow}(\rho) := \max_{\{  M_i \}} \left( S(\rho_A) - \sum_i p_i S(\rho_i) \right),
\end{equation}
where the maximization runs over POVMs $\{ M_i \}$ applied to Alice's system, $p_i := \tr(\rho( M_i \otimes \id))$ and $\rho_i := \tr_{A}(\rho(M_i \otimes \id))/p_i$ \cite{HV01}. In Ref. \cite{KW04} Koashi and Winter derived a beautiful relation between the entanglement of formation and the quantity given in Eq. (\ref{HV}), showing in particular the equivalence of the need of the limit in Eq. (\ref{distillablernadomness}) to the validity of Eq. (\ref{fulladditivity}) for every pair of channels.

An important simplification of the additivity problem, due to Shor \cite{Sho04}, shows that the additivity of the $\chi$-quantity is equivalent to a simpler question: the additivity of the minimum output entropy, defined as \cite{Sho04}
\begin{equation} \label{minoutput}
S_{\min}({\cal E}) := \min_{\rho} S({\cal E}(\rho)).
\end{equation}
It turns out that Eq. (\ref{fulladditivity}) holds true if, and only if, for every pair of channels ${\cal E}_1, {\cal E}_2$ 
\begin{equation} \label{additivityminoutputentropy}
S_{\min}({\cal E}_1 \otimes {\cal E}_2) =  S_{\min}({\cal E}_1) + S_{\min}({\cal E}_2).
\end{equation}

Recently, based on similar results on R\'enyi entropies by Winter \cite{Win07} and Hayden \cite{Hay07} (see also \cite{HW08, CN09a, CN09b}), Hastings  proved the breakthrough result that the minimum output entropy is not additive \cite{Has09}: in general, Eq. (\ref{additivityminoutputentropy}) does not hold true. This in turn implies that the limits in Eqs. (\ref{capacity}), (\ref{cost}), and (\ref{distillablernadomness}) are needed and thus that we are unfortunately further away from grasping these three capacities than what we might have expected.

Hastings argument combines the approach of Winter \cite{Win07} and Hayden \cite{Hay07} to the problem with powerful new ideas and 
techniques to construct randomized examples of channels violating Eq. (\ref{additivityminoutputentropy}). In particular, his argument is heavily based on an exact expression for the eigenvalue probability distribution of the reduced density matrix of a Haar distributed bipartite state \cite{LP88}. The main goal of the present paper is to revisit Hastings' proof by employing instead more general properties 
of the Haar distribution, such as large deviations bounds for the concentration of well-behaved functions around their mean-values in high dimensions.
This allows us to present the proof in a relatively concise form. Moreover, we will be able to strengthen slightly Hastings' result and 
prove non-additivity of the overwhelming majority of Haar random channels (for appropriate input, output, and environment dimensions). 
As a by-product, we also obtain a new result concerning the concentration of measure phenomenon in high dimensional quantum states, 
which may be of independent interest.

%Levy's lemma - a very general tool, and 
%a more specialized result of \cite{HHL04} where the concentration
%of a largest eigenvalue of a density matrix of Haar random bipartite 
%pure state around average is determined. We have found a version 
%of the proof, which uses solely the Levy's lemma, however the proof then looses
%much of its clarity. 

We would like to refer the reader to an earlier paper by Fukuda, King, and Moser of a similar spirit \cite{FKM09}, where Hastings' original argument is explained in great detail and rigor. In particular, the authors derived explicit lower bounds to the input, output and environment dimensions for which channels violating additivity can be constructed. Our approach is unlikely to provide better estimates than the ones found in Ref. \cite{FKM09}, as it does not rely on the exact probability distribution of the Schmidt coefficients of a Haar bipartite state. However, as our proof differs from the original in a few places, the optimization of the dimensions in our version of the proof may still be an interesting task (which we do not pursue here however).

\textbf{Notation:}
We denote the set of density matrices acting on a Hilbert space ${\cal H}$ by $D({\cal H})$. Moreover, we will often write $A$ and $B$ for finite dimensional Hilbert spaces, $A \otimes B$ or $AB$ for their tensor product, and $|A|, |B|$ for their dimensions. For a pure state $\ket{\psi^{AB}} \in AB$, we define $\psi^{AB} := \ket{\psi^{AB}}\bra{\psi^{AB}}$, while $\psi^A$ will denote $\tr_{B}(\psi^{AB})$, where $\tr_{B}$ is the partial trace over subsystem $B$. We denote the $d$-dimensional unitary group by $\mathbb{U}(d)$. We define the entropy deviation from its maximal value of a state $\rho \in {\cal D}(\mathbb{C}^d)$ by $\delta S(\rho) := \log(d) - S(\rho)$. Let $\mathbb{S}^{n} := \{ x \in \mathbb{R}^{n+1} : ||x||_2 = 1  \}$ denote the Euclidean sphere in $\mathbb{R}^{n+1}$ and $\mu$ denote the normalized rotationally invariant measure in $\mathbb{S}^{n}$ (the Haar measure). Finally, the Bachmann-Landau notation $g(n) = o(f(n))$ stands for $\forall k > 0, \exists n_0 : \forall n > n_0, \hspace{0.1 cm} g(n) \leq k f(n)$.

\textbf{Structure of the paper:} In section \ref{sec:defs} we present the main results of the paper as well as the key definitions used in the proofs. The counterexamples to the additivity conjecture are given by the combination of three propositions \ref{MainLemma1}, \ref{MainLemma2}, and \ref{largedeviationbound}, which are proven in sections \ref{proofmain1}, \ref{proofmainlenmma2}, and \ref{latgedeviationproof}, respectively. 

\section{Definitions and Main results}
\label{sec:defs}
We will consider channels from $A$ to $B$ of the form 
\begin{equation} \label{definitionE}
{\cal E}(\rho) = \tr_{A}\left(U\left(\rho^{A} \otimes \ket{0}\bra{0}^{B}\right)U^{\cal y}\right)
\end{equation}
for a unitary $U \in \mathbb{U}(|A||B|)$. The channels thus have input and environment dimensions equal to $|A|$ and output dimension equals to $|B|$. Moreover, we will make use the conjugate channel of ${\cal E}$, defined as
\begin{equation} 
\overline{{\cal E}}(\rho) = \tr_{A}\left(U^*\left(\rho^{A} \otimes \ket{0}\bra{0}^{B}\right)U^T\right).
\end{equation}

The counterexamples to the minimum output entropy additivity conjecture will be constructed by selecting the unitary $U$ at random from the Haar measure in $\mathbb{U}(|A| |B|)$ and considering the regime of a very large environment dimension $|A| \gg |B|$. 

%with $|A'| = |A|$ and $|B|/|A| = o(1)$. 
%We can associate with this channel its isometric Stinespring dilation $V : A' \hookrightarrow AB$ defined by $V \ket{\psi} = U(\ket{\psi}^{A'}\ket{0}^B)$, such that ${\cal E}(\rho) = \tr_{B}(V \rho V^{\cal y})$.  
Throughout the paper $c_0 > 0$ will denote a fixed constant which can be taken to be e.g. $c_0 = 1333$, while the Landau notation $o(1)$ will stand for a term which can be taken as small as desired by choosing $|A|$ large enough. Hastings theorem can be stated as follows.
\begin{theorem} \label{maintheorem}
For $U$ drawn from the Haar measure in $\mathbb{U}(|A| |B|)$, consider a channel as in Eq. (\ref{definitionE}). Then, for $c \geq c_0$, with probability $1 - o(1)$, 
\begin{equation} \label{eqmaintheorem}
S_{\min}({\cal E} \otimes \overline{{\cal E}}) \leq S_{\min}({\cal E}) + S_{\min}(\overline{{\cal E}}) - \frac{\log |B| - 2c}{|B|}.
\end{equation}
\end{theorem}

We will prove Theorem \ref{maintheorem} by the combination of two results. The first, analogous to a similar result of Winter and Hayden \cite{Win07, Hay07, HW08} on R\'enyi entropies, delivers an upper bound on the minimum output entropy of ${\cal E} \otimes \overline{{\cal E}}$ by considering the output entropy of the canonical maximally entangled state in $A \otimes B$ as an input.

\begin{lemma} \label{Hayden}
For a channel given by Eq. (\ref{definitionE}), 
\begin{equation}
S_{\min}({\cal E} \otimes \overline{{\cal E}}) \leq 2 \log |B| - \frac{\log |B| }{|B|}.
\end{equation}
\end{lemma}
For completeness, we reproduce the proof of Lemma (\ref{Hayden}) in Appendix \ref{Haydenproof}. 

The second result is a probabilistic argument for the existence of channels with high minimum output entropy. This is Hastings breakthrough contribution to the problem \cite{Has09}. 

\begin{lemma} \label{Hastings}
For $U$ drawn from the Haar measure in $\mathbb{U}(|A| |B|)$, consider a channel as in Eq. (\ref{definitionE}). Then, for $c \geq c_0$, with probability $1 - o(1)$, 
\begin{equation} \label{eqmaintheorem}
S_{\min}({\cal E}) \geq \log |B| - \frac{c}{|B|}.
\end{equation}
\end{lemma}

The main idea in the proof of Lemma \ref{Hastings} is to look at the probability that the output of a Haar random input state is close to a low entropy state (with entropy smaller than $\log|B| - c/|B|$). On one hand, we will show that for $|B|/|A| = o(1)$, this probability is upper bounded by $\exp(- c K |A|)$ (with $K > 0$ a constant), for a Haar random choice of the channel unitary. On the other hand, we will compute a lower bound on this probability, conditioned on the minimum output entropy of the channel being small; in this way we will get a lower bound of order $\exp( - \ln(2) |A|)$. Putting these two estimates together we obtain Lemma \ref{Hastings}. 

There are two key conceptual insights necessary to turn the idea of the previous paragraph into a proof. The first is to define an appropriate notion of closeness, when quantifying how close a state is to a low entropy one. For this, Hastings introduced the concept of a tube around a state \footnote{The name \textit{tube} is taken from Ref. \cite{FKM09}.}, which will take a central role in the proof of Lemma \ref{Hastings}.
\begin{definition}
We define the tube around $\sigma \in {\cal D}(\mathbb{C}^D)$ with width parameter $N > 0$ as
\begin{equation}
\text{TUBE}(\sigma, N) := \left \{ \pi \in {\cal D}(\mathbb{C}^D) : \exists \hspace{0.2 cm} \frac{1}{2} \leq p \leq 1 \hspace{0.3 cm} s.t. \hspace{0.3 cm} \left \Vert \pi - \left(p \sigma + (1 - p) \frac{\id}{D} \right) \right \Vert_{\infty} \leq \sqrt{\frac{\log(N)}{N}}  \right \}.
\end{equation}
\end{definition}
We will be interested in the probability that the output of a random input state, over a random choice of the channel, is in the tube of a low entropy state. The set of such states is formalized in the next definition.
\begin{definition}
For constants $N, c > 0$, we define the set of states in the \textit{tube} of a low entropy state as 
\begin{equation}
X_{D, N, c} := \left \{ \rho \in {\cal D}(\mathbb{C}^{D}) : \exists \hspace{0.2 cm} \sigma \in {\cal D}(\mathbb{C}^{D}) \hspace{0.2 cm} with \hspace{0.2 cm} \delta S(\sigma) \geq c/D \hspace{0.2 cm} s.t. \hspace{0.2 cm}\rho \in \text{TUBE}(\sigma, N) \hspace{0.2 cm}  \right\}.
\end{equation}
\end{definition}

The second insight is to consider the probability only of a particular subset of the set of states close to a low entropy state. We will look at the intersection of $X_{D, N, c}$ with the set of states of small operator norm. While this restriction will affect only very mildly the lower bound on the probability we are ultimately interesting in analyzing, it will allow us to get a much improved upper bound on it.
\begin{definition}
For a constant $a > 1$, we define the set of states with \textit{bounded} operator norm as
\begin{equation}
Y_{D, a} := \left \{ \rho \in {\cal D}(\mathbb{C}^{D})  : ||\rho||_{\infty} \leq \frac{a}{D} \right \}.
\end{equation}
\end{definition}

We are now in position to state precisely the two propositions which will be the focus of the remainder of the paper.  

Let $\ket{\chi} \in A$ be a Haar random state and ${\cal E}$ be a channel given by Eq. (\ref{definitionE}) with $U$ drawn from the Haar measure in $\mathbb{U}(|A| |B|)$. Then, for $|A| \geq |B|^2$, we have
\begin{proposition} \label{MainLemma1}
\begin{equation}
\Pr_{{\cal E}, \chi}\left( {\cal E}({\chi}) \in  X_{|B|, |A|, c} \cap Y_{|B|, a} \right) \leq \exp\left(-\frac{c|A|}{128a} + o(1)|A|\right).
\end{equation}
\end{proposition}

Moreover, for $\log |A| \geq 8|B|^8$ and $a \geq 15$,
\begin{proposition} \label{MainLemma2}
\begin{equation}
\Pr_{{\cal E}, \chi} \left( {\cal E}({\chi}) \in  X_{|B|, |A|, c} \cap Y_{|B|, a} \right) \geq \frac{1}{8 |A|}\exp(- \ln(2) |A|) \left(\Pr_{{\cal E}}\left( \delta S_{\min} \geq \frac{c}{|B|}  \right) - o(1) \right).
\end{equation}
\end{proposition}

Combining these two results we get Lemma \ref{Hastings} by choosing $c > 128 \ln(2) a$ and $a=15$. 

We will derive Proposition \ref{MainLemma1} from a new large deviation bound, which we believe might be of independent interest. It shows that with high probability the reduced state $\psi^A$ of a random bipartite state $\ket{\psi^{AB}}$ is close, in two norm, to the maximally mixed state. Although similar results are well-known (see e.g. \cite{HHL04, HLW06}), the restriction to reduced states $\psi^B$ with a small operator norm will allow us to sharpen the exponential bound essentially by a factor of $|B|$; this improvement turns out to be crucial in proving Proposition \ref{MainLemma1}. By a measure concentration argument we prove in section \ref{latgedeviationproof} the following
\begin{proposition} \label{largedeviationbound}
For $\ket{\psi^{AB}} \in A\otimes B$ drawn from the Haar measure, $|A| \geq |B|^2$ and $a \geq 3$, 
\begin{equation}
\Pr \left(\left \Vert \psi^B - \frac{\id}{|B|} \right \Vert_2 \geq  \varepsilon  \hspace{0.3 cm} \text{and} \hspace{0.3 cm}  \psi^B \in Y_{|B|, a}  \right) \leq 4 \exp\left(- \frac{|A| |B|^2  \left(\varepsilon - 2|A|^{-\frac{1}{2}}\right)^2}{64a} \right).
\end{equation}
\end{proposition}

%In the next three sections we prove Lemmata \ref{MainLemma1}, \ref{MainLemma2}, and \ref{largedeviationbound}.

\section{Proof of Proposition \ref{MainLemma1}} \label{proofmain1}

In this section we prove Proposition \ref{MainLemma1}. The idea is to combine Proposition \ref{largedeviationbound} and the following simple lemma relating the entropy deviation from its maximal value to the distance to the maximally mixed state.

\begin{lemma} \label{entversus2norm}
For every $\sigma \in {\cal D}(\mathbb{C}^D)$,
\begin{equation} \label{2normversusentropylemma2}
\left \Vert \sigma - \frac{\id}{D} \right \Vert_2^2 \geq \frac{\log(D) - S(\sigma)}{D}. 
\end{equation}
\end{lemma}
\begin{proof}
We have
\begin{eqnarray} \label{2normversusentropylemma1}
S(\rho) &\geq& - \log(\tr(\rho^2)) \nonumber \\
&=& - \log(D\tr(\rho^2)) + \log(D) \nonumber \\
&\geq& 1 - D\tr(\rho^2) + \log(D),
\end{eqnarray}
where the first inequality follows from the concavity of the $\log$ and the second from the relation $\log(x) \leq x - 1$, valid for $x \geq 1$. Rearranging 
terms in Eq. (\ref{2normversusentropylemma1}), we find Eq. (\ref{2normversusentropylemma2}).
\end{proof}

%Therefore, our first step in proving Lemma \ref{MainLemma1} will be the derivation of a strong enough large deviation bound on the probability that $|| \psi^B - \id/|B|||_2$ is big. We could follow Hastings in the original proof [] and proceed by using the exact probability distribution for the eingenvalues of $\psi^B$ []. We however will present an alternative derivation, based on measure concentration techniques. Here the restriction to reduced states $\psi^B$ with small operator norm will show crucial. The final result is the following.     

%Before we turn to the proof of Lemma \ref{largedeviationbound}, let us show how we can get Lemma \ref{MainLemma1} from it. 

\vspace{0.8 cm}

\begin{proof} (Proposition \ref{MainLemma1})

Let $\ket{\psi^{AB}}$ be such that $\psi^B \in X_{|B|, |A|, c}$. Then there is a $\sigma$ with $\delta S(\sigma) \geq c/ |B|$ such that $\psi^B \in \text{TUBE}(\sigma, |A|)$. From Lemma \ref{entversus2norm} we get 
\begin{equation} \label{prooflemma1main1}
\left \Vert \sigma - \frac{\id}{|B|} \right \Vert_2 \geq \frac{\sqrt{c}}{|B|}.
\end{equation}
As $|| \psi^B - (p \sigma + (1-p)\id/|B|)||_{\infty} \leq \sqrt{\frac{\log|A|}{|A|}}$, with $p \geq 1/2$, Eq. (\ref{prooflemma1main1}) gives
\begin{equation}%*
\left \Vert \psi^B - \frac{\id}{|B|}   \right \Vert_2 \geq \frac{\sqrt{c}}{2|B|} - \sqrt{\frac{\log |A|}{|A|}}. 
\end{equation}%*
where we used $||\cdot||_2\leq ||\cdot||_\infty$.

A moment of thought reveals that the distribution of ${\cal E}(\chi)$, for random ${\cal E}$ and $\chi$, is the same as the distribution of the reduced density matrix $\psi^B$ of a random bipartite state $\ket{\psi}^{AB} \in A\otimes B$. Therefore, from the argument of the previous paragraph
\begin{eqnarray}%*
\Pr_{{\cal E}, \chi} \left( {\cal E}({\chi}) \in  X_{|B|, |A|, c} \cap Y_{|B|, a} \right) &=& \Pr_{\psi}\left( \psi^B \in  X_{|B|, |A|, c} \cap Y_{|B|, a} \right)  \\ &\leq& \Pr_{\psi}\left(\left \Vert \psi^B - \frac{\id}{|B|}   \right \Vert_2 \geq \frac{\sqrt{c}}{2|B|} - \sqrt{\frac{\log |A|}{|A|}} \hspace{0.3 cm} \text{and} \hspace{0.3 cm}  \psi^B \in Y_{|B|, a} \right). \nonumber
\end{eqnarray}%*
The result now follows from Proposition \ref{largedeviationbound}.
\end{proof}

\section{Proof of Proposition \ref{MainLemma2}} \label{proofmainlenmma2}

On general lines, the idea of the lower bound given by Proposition \ref{MainLemma2} is the following. Let $P$ be the probability that a random channel has minimum output entropy bigger than $\log |B| - c/|B|$. For a given channel ${\cal E}$, let $\chi_{\cal E}$ be a pure input state to ${\cal E}$ with minimum output entropy, i.e. a state which satisfies $S({\cal E}(\chi_{\cal E})) = S_{\min}({\cal E})$. We will show that with probability larger than $\Omega(\exp( -\ln(2) |A|))$, ${\cal E}(\chi)$ is in the tube of ${\cal E}(\chi_{\cal E})$, for a random choice of the input state $\ket{\chi}$. From this we can conclude that ${\cal E}(\chi)$ is in the tube of a low entropy state with probability bigger than $(1 - P) \Omega(\exp( -\ln(2) |A|))$. 

This is almost all there is to show, except that from the argument of the previous paragraph, we have no guarantee that the states ${\cal E}(\chi)$ which we have proven to be in the tube of a low entropy state also belong to $Y_{|B|, a}$. To overcome this difficulty, we employ a large deviation bound due to Harrow, Hayden, and Leung \cite{HHL04} (Lemma \ref{HHLLemma} in Appendix \ref{HHLLemmasection}) which shows that with probability bigger than $1 - \exp(-|A|)$, ${\cal E}(\chi)$ belongs $Y_{|B|, a}$. This lemma thus allows us to disregard states not in $Y_{|B|, a}$ for sufficiently large $|A|$.   

%\textcolor{red}{
%I do not understand something here: we shall only use later that we can disregard 
%states not in $Y$ for large $|A|$. So, we shall not take $a$ 
%or $c$ large, but only $|A|$. So why in above estimate there is no $|A|$?}

%\textcolor{green}{It was a typo, it should be $|A|$ of course...}

\vspace{1 cm}

\begin{proof}
(Proposition \ref{MainLemma2})

The first step in the proof is to eliminate the event ${\cal E}(\chi) \in Y_{|B|, a}$. For this, we first use Lemma \ref{intersectionversusdiferenceinprob} of Appendix \ref{probabilityfacts} to get
\begin{eqnarray} \label{interbydifmain} 
\Pr_{{\cal E}, \chi} \left( {\cal E}({\chi}) \in  X_{|B|, |A|, c} \cap Y_{|B|, a} \right) &\geq& \Pr_{{\cal E}, \chi} \left( {\cal E}({\chi}) \in  X_{|B|, |A|, c}  \right) - \Pr_{{\cal E}, \chi} \left( {\cal E}({\chi}) \notin  Y_{|B|, a} \right).
\end{eqnarray}
Then, from Lemma \ref{HHLLemma} of Apendix \ref{HHLLemmasection},  
\begin{equation} \label{boundonNOTY}
\Pr_{{\cal E}, \chi} \left( {\cal E}({\chi}) \notin  Y_{|B|, a} \right) \leq \left( \frac{10 |B|}{a - 1}\right)^{2 |B|} \exp \left(- |A| \frac{(a - 1) - \log(a)}{14 \ln(2)}   \right) \leq \exp \left( -|A|\right),
\end{equation}
for $a \geq 15$ and $|A| \geq 2|B| \ln(2|B|)$. 

In the remainder of the proof we show that
\begin{equation} \label{boundonX}
\Pr_{{\cal E}, \chi} \left( {\cal E}({\chi}) \in  X_{|B|, |A|, c}  \right) \geq \frac{1}{8|A|}\exp(- \ln(2) |A|) \left(\Pr_{{\cal E}}\left( \delta S_{\min} \geq \frac{c}{|B|}  \right) - o(1) \right).
\end{equation}
The result then follows from Eqs. (\ref{interbydifmain}), (\ref{boundonNOTY}), and (\ref{boundonX}).

Let us define $\sigma_{\cal E} := {\cal E}(\chi_{\cal E})$, with $\chi_{\cal E}$ an input to ${\cal E}$ with minimum output entropy, i.e. a state such that $S({\cal E}(\chi_{\cal E})) = S_{\min}({\cal E})$. From the definition of the set $X_{|B|, |A|, c}$ we find
\begin{eqnarray} \label{relsigmaalmostoptimalprooflemma2}
\Pr_{{\cal E}, \chi} \left( {\cal E}({\chi}) \in  X_{|B|, |A|, c} \right) &\geq&  \Pr_{{\cal E}, \chi} \left( {\cal E}({\chi}) \in \text{TUBE}(\sigma_{\cal E}, |A|) \hspace{0.2 cm} \text{and} \hspace{0.2 cm} \delta S_{\min}({\cal E}) \geq \frac{c}{|B|}  \right).
\end{eqnarray}

We now proceed to bound the right-hand-side of Eq. (\ref{relsigmaalmostoptimalprooflemma2}). Following \cite{FKM09}, 
\begin{eqnarray} \label{expttrick} 
\Pr_{{\cal E}, \chi} \left( {\cal E}({\chi}) \in  \text{TUBE}(\sigma_{\cal E}, |A|) \hspace{0.2 cm} \text{and} \hspace{0.2 cm} \delta S_{\min}({\cal E}) \geq \frac{c}{|B|}  \right) &=& \nonumber \\ \mathbb{E}_{\cal E} \left( \textbf{1}\left(\delta S_{\min}({\cal E}) \geq \frac{c}{|B|}\right) \Pr_{\chi}\left( {\cal E}({\chi}) \in  \text{TUBE}({\cal E}(\sigma_{\cal E}, |A|) \right)\right), &&  
\end{eqnarray}
where $\textbf{1}(\delta S_{\min}({\cal E}) \geq c/|B|))$ is the indicator function of the event (only over channels):  $\{ \delta S_{\min}({\cal E}) \geq c/|B| \}$.  

Let us consider the probability over states inside the expectation value in Eq. (\ref{expttrick}). For a Haar random $\ket{\chi} \in A$, we can write
\begin{equation}%*
\ket{\chi} = \sqrt{x} \ket{\chi_{\cal E}} + \sqrt{1 - x} \ket{\phi},
\end{equation}%*
where $x = |\braket{\chi_{\cal E}}{\chi}|^2$ and $\ket{\phi}$ is a state orthogonal to $\ket{\chi_{\cal E}}$. In Lemma \ref{independetHaarmeasure} of Appendix \ref{probabilityfacts} we prove that $x$ and $\ket{\phi}$ are independent random variables and that $\ket{\phi}$ is distributed accordingly to the Haar measure in the subspace of $A$ orthogonal to $\ket{\psi}$. Therefore, 
\begin{eqnarray} \label{proboverchi}
\Pr_{\chi}\left( {\cal E}({\chi}) \in  \text{TUBE}({\cal E}(\sigma_{{\cal E}}), |A|) \right) \geq \Pr \left( x \geq 1/2 \right) \Pr_{\phi} \left( F \cap G \right),
\end{eqnarray}
where 
\begin{equation}%*
F := \left \{ \Vert {\cal E}(\ket{\chi_{\cal E}}\bra{\phi}) \Vert_{\infty} \leq \frac{1}{4}\sqrt{\frac{\log|A|}{|A|}} \right \}, \hspace{0.3 cm} G := \left \{ \left \Vert {\cal E}(\ket{\phi}\bra{\phi}) - \frac{\id}{|B|} \right \Vert_{\infty} \leq \frac{1}{2}\sqrt{\frac{\log|A|}{|A|}} \right \}.
\end{equation}%*
Indeed, note that if $x \geq 1/2$ and $F, G$ hold true
\begin{eqnarray}%*
\left \Vert {\cal E}(\chi) - x {\cal E}({\chi_{\cal E}}) - (1 - x)\frac{\id}{|B|} \right \Vert_{\infty} &\leq& 2 \Vert {\cal E}(\ket{\chi_{\cal E}}\bra{\phi}) \Vert_{\infty} + \left \Vert {\cal E}(\ket{\phi}\bra{\phi}) - \frac{\id}{|B|} \right \Vert_{\infty} \leq \sqrt{\frac{\log |A|}{|A|}},   
\end{eqnarray}%*
which implies ${\cal E}({\chi}) \in  \text{TUBE}({\cal E}(\sigma_{{\cal E}}), |A|)$. 

In Lemma \ref{geometric} we use a simple geometric argument to show 
\begin{equation}%*
\Pr\left( |\braket{\chi_{\cal E}}{\chi}|^2 \geq \frac{1}{2} \right) \geq \frac{1}{8|A|}\exp \left( - \ln(2) |A| \right).
\end{equation}%*
Then, from Eqs. (\ref{expttrick}) and (\ref{proboverchi})
\begin{eqnarray}%* 
&&\Pr_{{\cal E}, \chi} \left( {\cal E}({\chi}) \in  \text{TUBE}({\cal E}(\sigma_{{\cal E}}), |A|) \hspace{0.2 cm} \text{and} \hspace{0.2 cm} \delta S_{\min}({\cal E}) \geq \frac{c}{|B|}  \right)  \nonumber \\ &\geq& \frac{1}{8|A|} \exp \left( - \ln(2) |A| \right) \mathbb{E}_{\cal E} \left( \textbf{1}\left(\delta S_{\min}({\cal E}) \geq \frac{c}{|B|}\right) \Pr_{\chi}\left( F \cap G  \right)\right),  \nonumber \\ &=& \frac{1}{8|A|} \exp \left( - \ln(2) |A| \right) \Pr_{{\cal E}, \chi} \left( F \cap G \cap  \left( \delta S_{\min}({\cal E}) \geq \frac{c}{|B|} \right) \right).  
\end{eqnarray}%*

From Lemma \ref{intersectionversusdiferenceinprob} we can bound the second term in the last line of the equation above as
\begin{eqnarray}%* 
\Pr_{{\cal E}, \chi} \left( F \cap G \cap  \left( \delta S_{\min}({\cal E}) \geq \frac{c}{|B|} \right) \right) \geq  \Pr_{{\cal E}, \chi} \left( \delta S_{\min}({\cal E}) \geq \frac{c}{|B|} \right) - \Pr_{{\cal E}, \chi}\left(F^c \right) - \Pr_{{\cal E}, \chi}\left( G^c \right).  
\end{eqnarray}%*
Eq. (\ref{boundonX}) now follows from Lemma \ref{boundsmessy}, where we prove that $\Pr_{{\cal E}, \chi}(F^c), \Pr_{{\cal E}, \chi}(G^c) = o(1)$, asymptotically in $|A|$. 
\end{proof}

\begin{lemma} \label{geometric}
Let $\ket{\psi} \in A$ be a fixed state and $\ket{\chi} \in A$ be drawn from the Haar measure. Then, 
\begin{equation}%*
\Pr\left( |\braket{\psi}{\chi}|^2 \geq \frac{1}{2} \right) \geq 
\frac{1}{8|A|} \exp\left(-|A|\ln 2\right)
%\exp \left( - \ln(2) |A| \right).
\end{equation}%*
\end{lemma}
\begin{proof}

The vectors $\ket{\chi}$ can be seen as points $(x_1, \ldots, x_n)$
on real unit sphere $\mathbb{S}^{n-1}$ with $n = 2|A|$. The Haar measure is thus 
the normalized area of the sphere and the condition 
$|\braket{\psi}{\chi}|^2 \geq 1/2$ reads as  $x_1^2+x_2^2 \geq 1/2$.

Clearly $\Pr(x_1^2+x_2^2 \geq 1/2)$ is lower bounded by $\Pr(x_1^2\geq 1/2)$, 
which equals to the ratio of the area of a polar cap determined by the condition 
$x_1^2 \geq 1/2$ and the volume of the sphere. The area of the cap is in turn lower 
bounded by the volume of an $(n-1)$-dimensional ball given by the condition $x_2^2 + \ldots +x_n^2\leq 1/2$
(the projection of the cap onto a subspace perpendicular to the $x_1$ axis). Invoking explicit 
formulas for the volume of a ball and the area of a sphere (see e.g. \cite{MS86}), we obtain 
\be
\Pr(|\braket{\psi}{\chi}|^2 \geq 1/2) \geq \frac{1}{n \pi (\sqrt2)^{n-1}}
\geq \frac {1}{8|A|} e^{-\ln(2)|A|}.
\ee
\end{proof}

\begin{lemma} \label{boundsmessy}
Let $\ket{\psi}$ be a fixed state in $A$, $\ket{\phi}$ be drawn from the Haar measure in the subspace of $A$ orthogonal to $\ket{\psi}$ and ${\cal E}$ be a channel as in Eq. (\ref{definitionE}), with $U$ drawn from the Haar measure in $\mathbb{U}(|A||B|)$. Define 
\begin{equation}%*
F := \left \{ \Vert {\cal E}(\ket{\psi}\bra{\phi}) \Vert_{\infty} \leq \frac{1}{4}\sqrt{\frac{\log|A|}{|A|}} \right \}, \hspace{0.3 cm} G := \left \{ \left \Vert {\cal E}(\ket{\phi}\bra{\phi}) - \frac{\id}{|B|} \right \Vert_{\infty} \leq \frac{1}{2}\sqrt{\frac{\log|A|}{|A|}} \right \},
\end{equation}%*
Then, for $\log |A| \geq 8 |B|^8$ there are constants $C_1, C_2 > 0$ such that
\begin{equation} \label{boundsonFandG}
\Pr_{{\cal E}, \phi}\left( F \right) \geq 1 - \exp \left( - \frac{C_1 \log |A|}{|B|^8} \right), \hspace{0.2 cm} \Pr_{{\cal E}, \phi}\left( G \right) \geq 1 - \exp \left( - C_2 \log |A| \right).
\end{equation}
\end{lemma}
\begin{proof}
Let us start with the bound on the probability of $F^c$. Consider the complementary channel of ${\cal E}$, defined by ${\cal E}^c(\rho) := \tr_{B}\left(U\left(\rho^A \otimes \ket{0}\bra{0}^B \right)U^{\cal y}\right)$. Noting that ${\cal E}^c$ is a channel with input and output dimension $|A|$ and environment dimension $|B|$, we can write
\begin{equation}%*
{\cal E}^c(\rho) = \sum_{k=1}^{|B|^2} A_{k} \rho A_{k}^{\cal y},
\end{equation}%*
for Kraus operators $A_{k}$ such that $\sum_k A^{\cal y}_kA_k = \id$. Thus
\begin{equation}
{\cal E}(\rho) = \sum_{k=1}^{|B|^2} \sum_{k'=1}^{|B|^2} \tr(A_{k'}^{\cal y}A_k \rho) \ket{k}\bra{k'},
\end{equation}
from which we find
\begin{equation}
\left \Vert {\cal E}(\ket{\psi}\bra{\phi}) \right \Vert_{\infty} \leq |B|^4 \max_{k, k'} |\bra{\phi} A_{k'}^{\cal y}A_k \ket{\psi}|.
\end{equation}

Let $k_{\max}, k'_{\max}$ be the optimal indices in the equation above and define $\ket{\theta} := A_{k_{\max}'}^{\cal y}A_{k_{\max}} \ket{\psi} / \Vert A_{k_{\max}'}^{\cal y}A_{k_{\max}} \ket{\psi} \Vert^{1/2}$. As $||A_{k}||_{\infty} \leq 1$ for all $k$, we get $\Vert A_{k_{\max}'}^{\cal y}A_{k_{\max}} \ket{\psi} \Vert \leq 1$ and hence 
\begin{equation}
\left \Vert {\cal E}(\ket{\psi}\bra{\phi}) \right \Vert_{\infty} \leq |B|^4 |\braket{\theta}{\phi}|.
\end{equation}
We thus have
\begin{equation}
\Pr_{\phi}\left( F^c \right)\leq \Pr_{\phi} \left(  |\braket{\theta}{\phi}| \geq \frac{1}{4|B|^4}\sqrt{\frac{\log |A|}{|A|}} \right).
\end{equation}

Applying Lemma \ref{levysapplication} to the equation above we find
\begin{equation}
\Pr_{\phi} \left(  |\braket{\theta}{\phi}| \geq \frac{1}{4|B|^4}\sqrt{\frac{\log |A|}{|A|}} \right) \leq 2 \exp \left( - \frac{K \log |A|}{|B|^8} \right),
\end{equation}
for $\log |A| \geq 8|B|^8$ and a constant $K > 0$. This gives the bound on $\Pr(F)$ given in Eq. (\ref{boundsonFandG}).

Let us now turn to the bound on the probability of $G$. From Lemma \ref{independetHaarmeasure} of Appendix \ref{probabilityfacts}, we can select $\ket{\phi}$ by drawing $\ket{\chi} \in A$ from the Haar measure and setting $\ket{\chi} = \sqrt{x} \ket{\psi} + \sqrt{1 - x} \ket{\phi}$. Then we have
\begin{eqnarray}%*
\left \Vert {\cal E}(\phi) - \frac{\id}{|B|} \right \Vert_{\infty} &\leq& \left \Vert {\cal E}(\phi) - {\cal E}(\chi) \right \Vert_{\infty} + \left \Vert {\cal E}(\chi) - \frac{\id}{|B|} \right \Vert_{\infty} \nonumber \\ &\leq&  \left \Vert \phi - \chi \right \Vert_{1} + \left \Vert {\cal E}(\chi) - \frac{\id}{|B|} \right \Vert_{\infty},
\end{eqnarray}%*
where the first inequality follows from the triangle inequality and the second from the fact that $\Vert X \Vert_{\infty} \leq \Vert X \Vert_1$ and the monotonicity of the trace norm under trace preserving CP maps. Therefore, 
\begin{equation}
\Pr_{{\cal E}, \phi} \left( G \right) \geq \Pr_{{\cal E}, \chi} \left( \left \Vert \phi - \chi \right \Vert_{1} \leq \frac{1}{4}\sqrt{\frac{\log |A|}{|A|}} \hspace{0.2 cm} \text{and} \hspace{0.2 cm} \left \Vert {\cal E}(\chi) - \frac{\id}{|B|} \right \Vert_{\infty} \leq \frac{1}{4}\sqrt{\frac{\log |A|}{|A|}}  \right).
\end{equation}
From Lemma \ref{intersectionversusdiferenceinprob} of Appendix \ref{probabilityfacts}, in turn,  
\begin{equation} \label{lasteqofthislemma}
\Pr_{{\cal E}, \phi} \left( G \right) \geq 1 - \Pr_{{\cal E}, \chi} \left( \left \Vert \phi - \chi \right \Vert_{1} \geq \frac{1}{4}\sqrt{\frac{\log |A|}{|A|}}  \right) - \Pr_{{\cal E}, \chi} \left( \left \Vert {\cal E}(\chi) - \frac{\id}{|B|} \right \Vert_{\infty} \geq \frac{1}{4}\sqrt{\frac{\log |A|}{|A|}}  \right).
\end{equation}

One one hand, we have $\left \Vert \phi - \chi \right \Vert_{1} \leq \sqrt{2 - 2 |\braket{\phi}{\chi}|^2} = \sqrt{2x(2 - x)} \leq 2 \sqrt{x} = 2 |\braket{\psi}{\chi}|$. Following \cite{FKM09}, we find that if we replace $\ket{\phi}$ by $\ket{\chi}$, then with high probability it will only incur in a small error. Indeed, from Lemma \ref{levysapplication}
\begin{equation}%*
\Pr_{{\cal E}, \chi} \left( \left \Vert \phi - \chi \right \Vert_{1} \geq \frac{1}{4}\sqrt{\frac{\log |A|}{|A|}}  \right) \leq 2 \exp \left( - K \log |A| \right),
\end{equation}%*
for a constant $K > 0$. 

On the other hand, from Lemma \ref{HHLLemma} of section \ref{HHLLemmasection},
\begin{equation}%*
\Pr_{{\cal E}, \chi} \left( \left \Vert {\cal E}(\chi) - \frac{\id}{|B|} \right \Vert_{\infty} \geq \frac{1}{4}\sqrt{\frac{\log |A|}{|A|}}  \right) \leq \exp \left( - \frac{\log|A|}{560 \ln(2)} \right).
\end{equation}%*
Combining these two last equations with Eq. (\ref{lasteqofthislemma}), we find the lower bound on $\Pr(G)$ given in Eq. (\ref{boundsonFandG}).
\end{proof}

\begin{lemma} \label{levysapplication}
Let $S \subseteq A$ be a $|S|$-dimensional subspace of $A$ and let $P_S$ be the projector onto $S$. For $\ket{\phi} \in S$ drawn from the Haar measure in $S$ and a fixed $\ket{\theta} \in A$,
\begin{equation} 
\Pr_{\phi} \left(  |\braket{\theta}{\phi}| \geq \frac{1}{\sqrt{|S|}} + \varepsilon \right) \leq 4 \exp \left( - \frac{|S| \varepsilon^2}{16} \right),
\end{equation}
\end{lemma}
\begin{proof}
We prove the lemma by applying Levy's lemma, given in Lemma \ref{LevysLemma} of section \ref{latgedeviationproof}, with $f(\ket{\phi}) := |\braket{\theta}{\phi}|$. On one hand, we have 
\begin{equation}%*
\mathbb{E}\left( f(\ket{\phi})^2\right) = \bra{\theta}\left( \frac{P_S}{|S|} \right)\ket{\theta} \leq \frac{1}{|S|}. 
\end{equation}%*
Then, from the convexity of $x^2$, $\mathbb{E}\left( f(\ket{\phi}) \right)^2 \leq \mathbb{E}\left( f(\ket{\phi})^2\right) \leq |S|^{-1}$. On the other hand, the Lipschitz constant of $f$ is easily seen to be unity. The result then follows easily from Lemma \ref{LevysLemma}.
\end{proof}

\textbf{Remark:} We note that in the proof of Proposition \ref{MainLemma2} we set the input dimension $|A|$ to be exponentially larger than the 
output dimension $|B|$; this is due to the factor of $\log |A|/|A|$ in the definition of the tube. We could have instead defined 
the width of the tube as $f(|A|)/|A|$ for any function $f$ sublinear in $|A|$. In this way we can get a much better dependence
of the input dimension $|A|$ with the output dimension $|B|$. Besides that, as in Hastings' original proof, we have used equal input and environment dimensions.
However, our approach allow us to consider the general case in essentially the same fashion. In pricinple, this could lead to a better scaling of the
minimal dimensions for which counterexamples can be shown to exist.

\section{proof of Proposition \ref{largedeviationbound}} \label{latgedeviationproof}

%In the rest of this section we prove Lemma \ref{largedeviationbound}. Readers who would like to get a general idea of the proof, can skip to section \ref{proofmainlenmma2} in a first read. 

Let $\mathbb{S}^{n} := \{ x \in \mathbb{R}^{n+1} : ||x||_2 = 1  \}$ denote the Euclidean sphere in $\mathbb{R}^{n+1}$ and $\mu$ denote the normalized rotationally invariant measure in $\mathbb{S}^{n}$ (the Haar measure). Our strategy to prove Proposition \ref{largedeviationbound} is to explore the measure concentration phenomenon in high dimensional spheres \cite{MS86, Led01}. For a subset $A \subset \mathbb{S}^{n}$, define the $\varepsilon$-neighborhood of $A$ as
\begin{equation}%*
A_{\varepsilon} := \{ y \in \mathbb{S}^{n} : \exists \hspace{0.2 cm} x \in A \hspace{0.2 cm} s.t. \hspace{0.2 cm} || x - y ||_2 \leq \varepsilon \}.
\end{equation}%*

\begin{theorem} \label{measureceontration}
(Concentration of Measure in $\mathbb{S}^{n}$ \cite{MS86, Led01}) Let $A \subset \mathbb{S}^{n}$ and $0 \leq \epsilon \leq 1$. If $\mu(A) \geq 1/2$, then $\mu(A_{\varepsilon}) \geq 1 - 4 \exp\left( - \frac{(n+1) \epsilon^2}{16}  \right)$.
\end{theorem}

This theorem says that the area of $\mathbb{S}^n$ is sharply concentrated around any set with measure bigger than $1/2$. A simple but very powerful corollary of Theorem \ref{measureceontration} says that slowly varying functions on $\mathbb{S}^{n}$ attain a value very close to its average almost everywhere (see e.g. \cite{HLW06} for applications to quantum information theory). This is the content of Levy's Lemma. 

\begin{lemma} \label{LevysLemma} 
(Levy's Lemma \cite{MS86, Led01}) Let $f : \mathbb{S}^{n} \rightarrow \mathbb{R}$ be a function with Lipschitz constant $\eta$ and a point $x \in \mathbb{S}^{n}$ be chosen uniformly at random. Then  
\begin{equation}
\Pr \left( \left | f(x) - \mathbb{E}f \right | \geq \alpha  \right) \leq 4 \exp \left( - \frac{(n + 1)\alpha^2}{16\eta^2} \right).
\end{equation}
\end{lemma}
Given a Haar distributed state $\ket{\psi} \in A$, we can see it as an Haar distributed point in $\mathbb{S}^{2|A|-1}$. Therefore the lemma above applies to Haar pure states as well.  

The proof of Proposition \ref{largedeviationbound} will follow closely the standard argument for deriving Levy's Lemma (see e.g. \cite{MS86, Led01}). An important difference is that we are only interested in establishing a large deviation bound for a particular subset of the state space, namely for states $\ket{\psi^{AB}}$ whose the reduced state $\psi^B$ has operator norm bounded by $a/|B|$. Such a restriction will allow us to use an improved bound on the Lipschitz constant of the function $g(\ket{\psi^{AB}}) := \left \Vert \psi^B - \frac{\id}{|B|} \right \Vert_2$ and sharpen the exponential bound appearing in Levy's Lemma by a factor of $|B|/(4a)$.  

\vspace{0.2 cm}

\begin{proof} 
(Proposition \ref{largedeviationbound}) 
Define
\begin{equation}%*
g(\ket{\psi^{AB}}) := \left \Vert \psi^B - \frac{\id}{|B|} \right \Vert_2.
\end{equation}%*
Note that $g$ is a function from $\mathbb{S}^{2|A||B|-1}$ to $\mathbb{R}$. Let $m(g)$ be the median of $g$ and set $M := \{ \ket{\psi^{AB}} : g(\ket{\psi^{AB}}) \leq m(g) \}$. In Lemma \ref{median} we show $m(g) \leq 2 |A|^{-\frac{1}{2}}$. Thus for every $\ket{\psi^{AB}} \in M$, we have 
\begin{equation}%*
\Vert \psi^B \Vert_2^2 \leq \frac{1}{|B|} + m(g)^2 \leq \frac{1}{|B|} + \frac{4}{|A|}. 
\end{equation}%*
An application of Lemma \ref{entversusmaxeing1} of Appendix \ref{operatornormversusentropy} with $\lambda = a/|B|$ then gives the following bound on the operator norm of states in $M$,
\begin{equation}%*
\Vert \psi^B \Vert_{\infty} \leq \frac{3}{|B|} \leq \frac{a}{|B|},
\end{equation}%*
for every $\psi^{AB} \in M$ and $|A| \geq |B|^2$ and $a \geq 3$.

Consider a state $\ket{\psi^{AB}}$ such that 
\begin{equation} \label{eqproofmainfssdfdsf}
g(\ket{\psi^{AB}}) \geq m(g) + \beta \hspace{0.2 cm} \text{and} \hspace{0.2 cm} || \psi^B ||_{\infty} \leq a / |B|.
\end{equation}
Because of the bound on the operator norm of $\psi^{B}$, we can use Lemma \ref{improvedLipschitz} to find from the first inequality of Eq. (\ref{eqproofmainfssdfdsf}) that 
$\psi^{AB}$ must be at least $\beta \sqrt{\frac{|B|}{4a}}$ away from $M$. Furthermore, by definition of the median, $\mu(M) \geq 1/2$. Therefore from Theorem \ref{measureceontration}
\begin{equation}%*
\Pr \left(\left \Vert \psi^B - \frac{\id}{|B|} \right \Vert_2 \geq \varepsilon  \hspace{0.3 cm} \text{and} \hspace{0.3 cm}  \psi^B \in Y_{|B|, a}  \right) \leq 
1 - \mu \left( A_{(\epsilon - m(g)) \sqrt{|B|/4a}}   \right) \leq \exp\left(- \frac{|A| |B|^2  (\varepsilon - m(g))^2}{64a} \right),
\end{equation}%*
and we are done.
\end{proof}

The next lemma shows that for states with operator norm bounded by $a/B$, the Lipschitz constant of the function $g$ is improved by a factor of $\sqrt{|B|/(4a)}$.

\begin{lemma} \label{improvedLipschitz}
Let $\ket{\psi^{AB}}, \ket{\phi^{AB}} \in A \otimes B$ be such that $||\psi^B||_{\infty}, ||\phi^B||_{\infty} \leq a/|B|$. Then
\begin{equation}
\left | \left \Vert \psi^B - \frac{\id}{|B|} \right \Vert_2 - \left \Vert \phi^B - \frac{\id}{|B|} \right \Vert_2 \right | \leq \sqrt{\frac{4a}{|B|}} || \ket{\psi^{AB}} - \ket{\phi^{AB}} ||_2. 
\end{equation}
\end{lemma}
\begin{proof}
We assume without loss of generality that $\left \Vert \psi^B - \id/|B| \right \Vert_2 \geq \left \Vert \phi^B - \id/|B| \right \Vert_2$. Let $\{ \ket{i} \}_{i=1}^{\text{rank}(\psi^B)}$ be an eigenbasis for $\psi^B$ and define $M(\rho) := \sum_{i} \bra{i}\rho \ket{i} \ket{i}\bra{i}$. Then,
\begin{eqnarray} \label{Lip1}
\left | \left \Vert \psi^B - \frac{\id}{|B|} \right \Vert_2 - \left \Vert \phi^B - \frac{\id}{|B|} \right \Vert_2 \right | &=&  \left \Vert \psi^B - \frac{\id}{|B|} \right \Vert_2 - \left \Vert \phi^B - \frac{\id}{|B|} \right \Vert_2 \nonumber \\ &\leq& \left \Vert M(\psi^B) - \frac{\id}{|B|} \right \Vert_2 - \left \Vert M(\phi^B) - \frac{\id}{|B|} \right \Vert_2 \nonumber \\ &\leq & \Vert M(\psi^B) - M(\phi^B) \Vert_2,
\end{eqnarray}
where the first inequality follows from Lemma \ref{monotonicity2infty}, and the second inequality from the triangle inequality. 

Let $\{ p_k \}_{k}$ and $\{ q_k \}_{k}$ be the eigenvalues of $M(\psi^B) = \psi^B$ and $M(\phi^B)$, respectively. Since $||\psi^B||_{\infty}, || \phi^B ||_{\infty} \leq a/|B|$, we find from Lemma \ref{monotonicity2infty} that $(\max_{k}p_k), (\max_{k}q_k) \leq a/|B|$. Hence
\begin{eqnarray} \label{Lip2}
\Vert M(\psi^B) - M(\phi^B) \Vert_2^2 &=& \sum_k (p_k - q_k)^2 \\ &=& \sum_k (\sqrt{p_k} - \sqrt{q_k})^2(\sqrt{p_k} + \sqrt{q_k})^2 \nonumber \\ &\leq& \frac{4a}{|B|} \sum_k (\sqrt{p_k} - \sqrt{q_k})^2 \nonumber \\ &=& \frac{4a}{|B|}\left( 2 - 2 F(M(\psi^B), M(\phi^B))  \right) \nonumber \\ &\leq&  \frac{4a}{|B|}\left( 2 - 2 F(\psi^B, \phi^B)  \right) \nonumber \\  &\leq&  \frac{4a}{|B|}\left( 2 - 2 F(\psi^{AB}, \phi^{AB})  \right) = \frac{4a}{|B|} || \ket{\psi^{AB}} - \ket{\phi^{AB}} ||_2^2, \nonumber
\end{eqnarray}
where the last two inequalities follows from the monotonicity of the fidelity under trace preserving CP maps. Putting Eqs. (\ref{Lip1}) and (\ref{Lip2}) together gives the result.
\end{proof}

The next lemma gives an upper bound on the median of the function $g$. 

\begin{lemma} \label{median}
Let $g : \mathbb{S}^{2|A||B|} \rightarrow \mathbb{R}$ be such that
\begin{equation}%* 
g(\ket{\psi^{AB}}) := \left \Vert \psi^B - \frac{\id}{|B|} \right \Vert_2
\end{equation}%*
and $m(g)$ be the median of $g$. Then $m(g) \leq 2|A|^{-\frac{1}{2}}$.
\end{lemma}
\begin{proof}
We start by bounding the median by the expectation value of $g$ as follows
\begin{equation}%*
\mathbb{E}g = \int_{g \geq m(g)} g(\psi) \mu(d \psi) + \int_{g \leq m(g)} g(\psi) \mu(d \psi) \geq m(g) \int_{g \geq m(g)} \mu(d \psi) = \frac{m(g)}{2}.
\end{equation}%*
We proceed by lower bounding the expectation value of $g(\ket{\psi})$, 
\begin{equation} \label{expect1}
(\mathbb{E} g)^2 \leq \mathbb{E}(g^2) = \mathbb{E}\left( \tr((\psi^B)^2)\right) - \frac{1}{|B|} = tr\left( \mathbb{E}\left( \psi^{AB} \otimes \psi^{A'B'}  \right) \id_{AA'} \otimes \mathbb{F}^{BB'}\right) - \frac{1}{|B|},
\end{equation}
where $\mathbb{F}^{BB'}$ is the swap operator the two systems $BB'$. The first inequality of the equation above follows from the convexity of $x^2$. From Schur's Lemma,
\begin{equation} \label{expect2}
\mathbb{E} \left( \psi^{AB} \otimes \psi^{A'B} \right) = \frac{\id^{AA'BB'} + \mathbb{F}^{AA'}\otimes \mathbb{F}^{BB'}}{|A||B|(|A||B| + 1)}.
\end{equation}
Putting Eqs. (\ref{expect1}) and (\ref{expect2}) together gives $m(g) \leq 2|A|^{-\frac{1}{2}}$.  
\end{proof}

The final lemma of this section shows the monotonicity of the operator and two norms under pinching. 

\begin{lemma} \label{monotonicity2infty}
For every $X$,
\begin{equation}
\Vert X \Vert_2 \geq \left \Vert \sum_k P_k X P_k \right \Vert_2, \hspace{0.2 cm} \Vert X \Vert_{\infty} \geq \left \Vert \sum_k P_k X P_k  \right \Vert_{\infty},
\end{equation}
for orthogonal projectors $P_k$ with $\sum_k P_k$ = 1.
\end{lemma}
\begin{proof}
Direct calculation. 
\end{proof}

\section{Acknowledgement}
We thank Robert Alicki for sharing his analysis of Hastings' paper (which triggered us to undertake a similar study) and Graeme Smith for providing the proof of Lemma \ref{independetHaarmeasure}. This work was supported by EC IP SCALA and an EPSRC Postdoctoral Fellowship for Theoretical Physics. FB would like to thank the hospitality of the members of the National Quantum Information Centre of Gda\'nsk, where this work was done. 

\appendix                                                                    

\section{A few probability facts} \label{probabilityfacts}

\begin{lemma} \label{intersectionversusdiferenceinprob}
For two events $M, N$, $\Pr(M \cap N) \geq \Pr(M) - \Pr(N^c)$, where $N^c$ is the complement of $N$.
\end{lemma}
\begin{proof}
We have 
\begin{equation}%*
\Pr(M) = \Pr(M \cap N) + \Pr(M \cap N^c) \leq Pr(M \cap N) + \Pr(N^{c}).
\end{equation}%*
Rearranging terms in the equation above gives the result of the lemma.
\end{proof}

\begin{lemma} \label{independetHaarmeasure}
Let $\ket{\chi} \in A$ be drawn from the Haar measure. Write
\begin{equation}%*
\ket{\chi} = \sqrt{x} \ket{\psi} + \sqrt{1 - x} \ket{\phi},
\end{equation}%*
where $\ket{\psi} \in A$ is a fixed state, $x = |\braket{\psi}{\chi}|^2$, and $\ket{\phi}$ is a state orthogonal to $\ket{\psi}$. Then $x$ and $\ket{\phi}$ are independent random variables and $\ket{\phi}$ is distributed accordingly to the Haar measure in the subspace of $A$ orthogonal to $\ket{\psi}$. 
\end{lemma}
\begin{proof}
Let $p_{A}(\ket{\psi})$ be the probability density function associated with the Haar measure in $A$. We can write $p_{A}(\ket{\psi}) = p_{A}(x, \ket{\phi})$. From the invariance of the Haar measure under unitary transformations, $p_{A}(U\ket{\psi}) = p_{A}(x, U\ket{\phi})$, for every $x$ and every unitary $U$ which acts non-trivially only in the subspace of $A$ orthogonal to $\ket{\psi}$, $A_{\psi^{\bot}}$. Therefore, the conditional probability density function
\begin{equation}%*
p_{A}(\ket{\phi} \hspace{0.05 cm} | \hspace{0.05 cm} x) = \frac{p_{A}(\ket{\phi}, x)}{p_{A}(x)}
\end{equation}%* 
is such that $p_{A}(U\ket{\phi} \hspace{0.05 cm}|\hspace{0.05 cm} x) = p_{A}(\ket{\phi} \hspace{0.05 cm}|\hspace{0.05 cm} x)$ for every $x$ and unitary $U$ acting on $A_{\psi^{\bot}}$. From the uniqueness of the Haar measure, we find that for every $x$, $p_{A}(\ket{\phi} \hspace{0.05 cm}|\hspace{0.05 cm} x) = p_{A_{\psi^{\bot}}}(\ket{\phi})$. This shows both that $\ket{\phi}$ is independent of $x$ and that it is Haar distributed.     
\end{proof}

\section{Relating operator norm, two norm, and entropy} \label{operatornormversusentropy}

\begin{lemma} \label{entversusmaxeing1}
Let $\rho \in{\cal D}(\mathbb{C}^D)$ be such that $\Vert \rho \Vert_{\infty} \geq \lambda > 1/D$. Then 
\begin{equation}
S(\rho) \leq s(\lambda, D) := (1 - \lambda) \log(D-1) + h(\lambda),
\end{equation}
and
\begin{equation}
\Vert \rho \Vert_{2}^2 \geq \lambda^2 + \frac{(1 - \lambda)^2}{D - 1}, 
\end{equation}
where $h(x) := - x \log x - (1 - x)\log(1 - x)$ is the Shannon binary entropy. 
\end{lemma}
\begin{proof}
Let $\lambda_i$ be the eigenvalues of $\rho$ in decreasing order. Then, for every $N \in \{1, ..., D \}$, 
\begin{equation}
\sum_{i=1}^N \lambda_i \geq \lambda_1 + (N-1)\frac{(1 - \lambda_1)}{D - 1},
\end{equation}
which shows that $\{ \lambda_i \}$ is majorized by the probability distribution $q := \left \{  \lambda_1, \frac{(1 - \lambda_1)}{D - 1}, ..., \frac{(1 - \lambda_1)}{D - 1} \right \}$. From the Schur convexity of $x \log x$, 
\begin{equation}
S(\rho) \leq S(q) = s(\lambda_1, D) := (1 - \lambda_1) \log(D-1) + h(\lambda_1). 
\end{equation}
A simple calculation shows that $\frac{\partial s(\mu)}{\partial \lambda} \leq 0$ for all $\mu \geq 1/D$. Therefore, the function $s(\lambda, D)$ is monotonic decreasing in $\lambda$ for $\lambda \geq 1/D$. As $\lambda_1 = \Vert \rho \Vert_{\infty} \geq \lambda$, we find that $S(\rho) \leq s(\lambda, D)$. 

The bound on the two norm can be obtained in an analogous way. As $x^2$ is Schur convex, we get that 
\begin{equation}
\Vert \rho \Vert_{2}^2 \geq \Vert q \Vert_2^2 = r(\lambda_1, D) := \lambda_1^2 + \frac{(1 - \lambda_1)^2}{D - 1}.
\end{equation}
A simple calculation shows that $r(\lambda_1, D)$ is monotonic increasing in $\lambda_1$, so that $r(\lambda_1, D) \geq r(\lambda, D)$. 
\end{proof}

\section{Large deviation bound for the operator norm} \label{HHLLemmasection} 

The following lemma, due to Harrow, Hayden, and Leung \cite{HHL04} is used twice in the proof of Proposition \ref{MainLemma2}. 

\begin{lemma} \label{HHLLemma}
(Lemma III.4 of \cite{HLW06}) Let $\ket{\psi^{AB}} \in A \otimes B$ be drawn from the Haar measure. For every $0 < \varepsilon < 1$,
\begin{equation}
\Pr_{\psi}\left( \left \Vert \psi^B \right \Vert_{\infty} \geq \frac{1}{|B|} + \frac{\varepsilon}{|B|}  \right) \leq  \left( \frac{10 |B|}{\varepsilon}\right)^{2 |B|} \exp \left(- |A| \frac{\varepsilon^2}{14 \ln(2)}   \right),
\end{equation}
while for every $\varepsilon > 0$ \cite{HLS05}
\begin{equation}
\Pr_{\psi}\left( \left \Vert \psi^B \right \Vert_{\infty} \geq \frac{1}{|B|} + \frac{\varepsilon}{|B|}  \right) \leq  \left( \frac{10 |B|}{\varepsilon}\right)^{2 |B|} \exp \left(- |A| \frac{\left(\varepsilon - \log(1 + \varepsilon)\right)}{14 \ln(2)}   \right),
\end{equation}

\end{lemma}

\section{Proof of Lemma \ref{Hayden}} \label{Haydenproof}

Following Refs. \cite{Hay07, HW08}, we use the canonical maximally entangled state $\ket{\Phi^{AA'}} := |A|^{-1}\sum_{i=1}^{|A|} \ket{i}^A\ket{i}^{A'}$ as an input state to 
\begin{equation}
{\cal E} \otimes \overline{{\cal E}}(\rho) = \tr_{AA'}\left((U \otimes U^*)\left(\rho^{AA'} \otimes \ket{0}\bra{0}^{B} \otimes \ket{0}\bra{0}^{B'} \right)(U \otimes U^*)^{\cal y}\right),
\end{equation}
where $U$ acts on $AB$ and $U^*$ on $A'B'$.

We can get a lower bound on the operator norm of ${\cal E} \otimes \overline{{\cal E}}(\Phi^{AA'})$ as follows
\begin{eqnarray}
\left \Vert {\cal E} \otimes \overline{{\cal E}}(\Phi^{AA'}) \right \Vert_{\infty} &\geq& \tr\left(\Phi^{BB'}{\cal E} \otimes \overline{{\cal E}}(\Phi^{AA'})\right) \nonumber \\ 
&=& \tr\left(\Phi^{BB'}\tr_{AA'}\left((U \otimes U^*)\left(\Phi^{AA'} \otimes \ket{0}\bra{0}^{B} \otimes \ket{0}\bra{0}^{B'} \right)(U \otimes U^*)^{\cal y}\right)\right) \nonumber \\ &=& \tr\left(\id^{AA'} \otimes \Phi^{BB'}\left((U \otimes U^*)\left(\Phi^{AA'} \otimes \ket{0}\bra{0}^{B} \otimes \ket{0}\bra{0}^{B'} \right)(U \otimes U^*)^{\cal y}\right)\right) \nonumber \\ &\stackrel{(i)}{\geq}& \tr\left(\Phi^{AA'} \otimes \Phi^{BB'}\left((U \otimes U^*)\left(\Phi^{AA'} \otimes \ket{0}\bra{0}^{B} \otimes \ket{0}\bra{0}^{B'} \right)(U \otimes U^*)^{\cal y}\right)\right) \nonumber \\ &=& \tr\left((U \otimes U^*)^{\cal y}\Phi^{AA'}\otimes \Phi^{BB'}(U \otimes U^*)\left(\Phi^{AA'} \otimes \ket{0}\bra{0}^{B} \otimes \ket{0}\bra{0}^{B'} \right)\right) \nonumber \\ &\stackrel{(ii)}{=}& \tr\left(\Phi^{AA'}\otimes \Phi^{BB'}\left(\Phi^{AA'} \otimes \ket{0}\bra{0}^{B} \otimes \ket{0}\bra{0}^{B'} \right)\right) \geq \frac{1}{|B|}.
\end{eqnarray}
In $(i)$ we used $\Phi^{BB'} \leq \id$, while $(ii)$ follows from the identity $\left(\id^C \otimes X^{C'}\right) \ket{\Phi^{CC'}} = \left( \left(X^C\right)^T \otimes \id^{C'}\right) \ket{\Phi^{CC'}}$.

Applying Lemma \ref{entversusmaxeing1} to ${\cal E} \otimes \overline{{\cal E}}(\Phi^{AA'})$, with $D = |B|^2$ and $\lambda = |B|^{-1}$ then gives
\begin{equation}%*
S \left({\cal E} \otimes \overline{{\cal E}}(\Phi^{AA'}) \right) \leq s(|B|^{-1}, |B|^2) = 2 \log|B| - \frac{\log |B|}{|B|}.
\end{equation}%*

\end{document}